\documentclass[12pt,a4paper]{article}
\usepackage{authblk}

\usepackage{graphicx}
\usepackage{hyperref}
\usepackage{amssymb,amsmath,amstext,amsfonts}

\usepackage{a4wide}

\usepackage{fullpage}
\usepackage{amsfonts,amsmath,amssymb,amsthm,boxedminipage,color,url,fullpage}
\usepackage{enumerate,paralist}
\usepackage[numbers]{natbib} 
\usepackage[labelfont=bf]{caption}
\usepackage{aliascnt,cleveref}

\usepackage{breqn}

\usepackage{epic,eepic}

\newtheorem{lemma}{Lemma}[section]

\newtheorem{theorem}[lemma]{Theorem}
\newtheorem{corollary}[lemma]{Corollary}

\newtheorem{claim}[lemma]{Claim}

\theoremstyle{definition}

\theoremstyle{remark}

\hypersetup{colorlinks=true,linkcolor=blue,citecolor=blue,filecolor=blue,urlcolor=blue}

\title{The price of anarchy and stability in general noisy best-response dynamics\footnote{Part of this work has been done while at LIAFA,  Universit\'e Paris Diderot, supported by the French ANR Project DISPLEXITY.}}

\author[1]{Paolo Penna}
\affil[1]{Department of Computer Science. ETH Zurich, Switzerland\\
  \href{mailto:paolo.penna@inf.ethz.ch}{\texttt{paolo.penna@inf.ethz.ch}}}

\usepackage{amssymb}

\usepackage[numbers]{natbib}
\usepackage{amsmath}
\usepackage{graphicx}

\usepackage{multirow,xspace}

\newcommand{\deltajobs}[1]{\underbrace{\delta,\ldots,\delta}_{#1}}
\newcommand{\Deltajobs}[1]{\underbrace{\Delta,\ldots,\Delta}_{#1}}

\newcommand{\PoA}{\textsc{PoA}\xspace}
\newcommand{\PoS}{\textsc{PoS}\xspace}

\newcommand{\ind}{\textsc{independent-logit}}
\newcommand{\asy}{\textsc{logit}}

\renewcommand{\citet}{\cite}

\begin{document}

\maketitle

\begin{abstract}
Logit-response dynamics \citep{AloNet10} are a rich and natural class of \emph{noisy} best-response dynamics.  In this work we revise the price of anarchy and the price of stability by considering the quality of long-run equilibria in these dynamics. Our results show that prior studies on simpler dynamics of this type can strongly depend on a \emph{synchronous} schedule of the players' moves.  In particular, a small noise by itself is not enough to improve the quality of equilibria as soon as other very natural schedules are used. 
 \end{abstract}

\section{Introduction}
Complex and distributed systems are often modeled by means of game dynamics in which agents (players) act spontaneously, and typically 
attempt to maximize their own benefit.  These dynamics can be regarded as distributed algorithms that are implemented by the players and have applications in Economics (e.g., markets), Physics (e.g., Ising model and spin systems), Biology (e.g., evolution of life), and in Computer Science (e.g., distributed protocols for routing and resource allocations by competing entities). As a result of such selfish behavior the system will reach some equilibrium that may or may not be optimal.
The \textsc{price of anarchy (\PoA)} \citep{KouPap09} and the \textsc{price of stability (\PoS)} \citep{AnsDasKleTarWexRou04}  have been introduced to quantify the efficiency loss caused by the players selfish behavior. Intuitively, the price of anarchy compares the \emph{worst} possible equilibrium to the optimum (the centrally designed outcome). Similarly, the price of stability compares the \emph{best} possible equilibrium with the optimum (to capture the loss caused by the inherent instability of optimal outcomes).  
The \PoA and the \PoS measure the performance of ``decentralized and anarchic'' distributed systems in which we simply let the players decide by themselves how to play, until they reach some equilibrium.

These two concepts seem naturally related to the question of explaining \emph{how} the players can choose between different equilibria. 
A remarkable observation made in game theory is that \emph{noise} allows the players to select good equilibria via a simple and natural dynamics. Also, results on the \PoA and \PoS for certain problems are sometimes based on equilibria that are ``unnatural'' because it is difficult for the players to reach or to remain in such equilibria.  This  leads naturally to reconsider the bounds on the \PoA and \PoS by restricting attention to ``reasonable'' equilibria. 

In this work we follow this approach by considering the equilibria that arise in certain natural \emph{noisy best-response} dynamics. In these dynamics, players keep updating their strategies and, at low noise regime, they tend to select strategies that deliver higher payoffs (see Section~\ref{sec:prely} for the formal model). The equilibria that are selected in the long-run are the so-called \emph{stochastically stable} states. That is, the states that in the stationary distribution have positive probability as the noise vanishes.\footnote{In this work we consider  \emph{pure} Nash equilibria, that is, states in which each player chooses one strategy and unilateral deviations are not beneficial. In general games, stochastically stable states are not necessarily Nash equilibria, even for the class of potential games (which is the main focus of this work).}  One of the most popular and studied of such dynamics is the following one: 
\begin{quote}
	\emph{Logit dynamics \citep{Blu93}.} At each time step  \emph{one randomly chosen player} updates her strategy according to a noisy best-response rule.
\end{quote}
On the one hand, for potential games, the long-run outcomes can be simply characterized   as the Nash equilibria \emph{minimizing the potential} of the game \cite{Blu98}.
On the other hand these dynamics seem to require a particular \emph{schedule} of the players' moves (who is allowed to update at each step and who stays put).  The authors in \citet{AloNet10} analyze the generalization in which  players' moves are not necessarily ``synchronized'':
\begin{quote}
	\emph{Independent learning logit-response dynamics \citep{AloNet10}.} At each time step a \emph{subset} of randomly chosen players update their strategies according to a noisy best response-rule (the same rule of logit dynamics). Every subset has some positive probability of being selected (e.g., when every player independently decides to revise his/her strategy with some probability $p\in (0,1)$).
\end{quote}
While the characterization of \cite{Blu98} in terms of potential minimizers is no longer true, \citet{AloNet10}  characterize long-run outcomes as the states minimizing a \emph{stochastic potential}. 
This means that the set of equilibria selected by the two dynamics is \emph{different} \cite{AloNet10}.  It is therefore natural to ask to what extent the  \emph{schedule} of the players' moves (the \emph{revision process}, in the game theory terminology) can affect the \emph{quality} of the selected equilibria in these dynamics. The logit-response dynamics has been proven to maintain many of the nice properties of logit dynamics, including the selection of good equilibria in coordination games, or convergence to Nash equilibria under certain conditions \cite{AloNet10,MarSha12,CauDurGauTou14,AloNet15}.

\subsection{Our Contribution}
This work presents a refined version of the quality of equilibria in games (\PoA and \PoS) based on `decentralized' dynamics with noise, that is, the \emph{logit-response with independent learning} described above. This approach has two advantages: 
\begin{enumerate}
	\item It suggests an intuitive definition of `reasonable' equilibria (robust to small mistakes and easy to reach without specific schedules of players' moves) to quantify the efficiency loss. 
	\item It allows to understand the role of \emph{synchronization} in such dynamics (prior studies considered the quality of equilibria in `synchronous' logit dynamics). 
\end{enumerate} 
Specifically, we propose to study a natural 
concept of \ind~\PoA (and \PoS) which compares the worst (respectively, best) long-run outcomes of independent learning logit-response dynamics with the optimum. That is, the classical \PoA (respectively, \PoS) restricted to the set of stochastically stable states \cite{AloNet10}. The analogous notions for logit dynamics, say the \asy~\PoA and \asy~\PoS, have been already studied in the literature under different names.\footnote{\citet{AsaSab09} call the \PoA restricted to potential minimizers \textsc{Inefficiency Ratio of Stable Equilibria (IRSE)}, while \citet{KawMak13} use the terms \textsc{Potential Optimal Price of Anarchy (PO\PoA)} and \textsc{of Stability (POPoS)}. In this work we prefer to use the terms \asy~\PoA and \asy~\PoS to emphasize the comparison between logit and independent-learning logit-response dynamics.} A direct comparison between these notions is informative of the importance of the schedule of the moves.

We study two classes of problems that have different features: \emph{load balancing} games which provide a simple model of congestion, and \emph{broadcast network design} games in which players share the cost of the used resources. Intuitively, in the first case players dislike to choose the same resources (because this creates  congestion and players experience  some delay), while in the second case they like to do the opposite (because the more players use a resource the less each of them pays for it).

We first consider \emph{load balancing} games whose analysis for logit dynamics has been done by \citet{AsaSab09}.
We obtain the tight bounds  shown in Table~\ref{tbl:results} which  illustrate of how different the two dynamics can perform. 
In particular, though the prior bounds on logit dynamics show that the quality of equilibria \emph{improves} compared to ``generic'' Nash equilibria, the opposite happens for independent-learning dynamics. 
This tells us that the improvement of the quality of the equilibria is due to the \emph{combination} of the players' response rule and of the possibility of \emph{synchronizing} their moves. Thus, the bounds based on potential minimizers are in general not ``robust'' in the sense that they depend critically on this assumption.
Also, contrary to logit dynamics, the stochastically stable states are not necessarily Nash equilibria (in fact they are far from Nash equilibria in terms of cost).

\begin{table} 
	\centering
	\begin{tabular}{|lr|c|c|}\hline
		\multicolumn{2}{|c|}{Noise with Synchronization} & Deterministic & Noise without Synchronization\\
		\hline\hline
		\multicolumn{4}{|c|}{Worst Equilibrium}\\
		\multicolumn{2}{|c|}{\asy~\PoA \cite{AsaSab09}} & \PoA  \cite{FinHor79} & \ind~\PoA \\ 
		$\geq 19/18$ & $\leq 3/2$  & $2\Big(1 - \frac{1}{m+1}\Big)$ &  $m$ \\	
		\hline
		\hline
		\multicolumn{4}{|c|}{Best Equilibrium}\\
		\multicolumn{2}{|c|}{\asy~\PoS}  & \PoS & {\ind~\PoS}\\
		\multicolumn{2}{|c|}{1} & 1 & {$2\Big(1 - \frac{1}{m+1}\Big)$} \\
		\hline
	\end{tabular} 
	\caption{Bounds for load balancing games: The rightmost column contains our contributions, and $m$ denotes the number of machines. The bounds on (\asy) \PoS are folklore.}
	\label{tbl:results}
\end{table}

We then consider broadcast \emph{network design} games studied in \cite{KawMak13,MamMih15}  for potential minimizers (logit dynamics), and show that in this case synchronization is not essential to obtain optimal configurations. First, for the case of  parallel links between a common source and a common destination (links represent resources of different costs) players will select the optimum if there is a unique shortest link (Theorem~\ref{th:broadcast:identical-locations}). Note that in these instances the price of anarchy is rather high ($\PoA=n$), while our result says that the restriction to stochastically stable states is instead optimal ($\ind~\PoA=1$). Finally, we show that the result does not extend to different sources since in some instances $\ind~\PoA>\asy~\PoA$ (Theorem~\ref{th:network-design-games}). Also in this case stochastically stable states are not Nash equilibria.

\subsection{Further Related Work}
The impact of the schedule of the players' moves on the quality of the equilibria has been studied earlier. The approach by \citet{MarSha12} is closely related to the one here, and it consists of deriving sufficient conditions for which the stochastically stable states of  logit-response dynamics  are the same as those of logit dynamics (potential minimizers). 
The impact of simultaneous moves on (deterministic) best-response dynamics is studied by \citet{FanMosSko12}. The performance of various \emph{concurrent} dynamics in which all players are active at each time step is studied in e.g. \citet{FotKapSpi10,AckBerFisHoe09,KlePilTar09,AulFerPasPenPer13}.

Several works refine the \PoA and the \PoS by restricting to certain type of equilibria. The \emph{price of stochastic anarchy} by \citet{ChuKatPhrRot08,ChuPyr09} is similar to ours in the sense that it considers stochastically  stable states though for a (concurrent) imitation dynamics.  
The authors in \citet{AsaSab09} were the first to restrict to potential minimizers (the long-run equilibria of logit dynamics)  and their analysis on atomic congestion games and load balancing games shows that prior bounds on the \PoA for these games can be too pessimistic. The work \citet{KawMak13} extends this approach to the \PoS and shows similar results for  broadcast network design games (where \textsc{logit PoS = logit PoA} up to a constant factor). The  restriction to the ring topology is considered in \citet{MamMih15} (together with other variants of \PoS and \PoA).
The authors in \citet{AulFerPasPer13} suggest to consider the expected social cost in logit dynamics and therefore the stationary distribution  as the equilibrium concept.  A new natural dynamics converging to  Nash equilibria with optimal social welfare in general games has been proposed in \citet{PraPey12}.
The robustness of the \PoA and \PoS bounds has been investigated by \citet{ChrKouSpi11} and \citet{Rou09} who considered approximate and correlated equilibria, respectively.

Stochastic stability is the tool to show that in the coordination game players select the risk dominant strategy (see \cite{AloNet10} for discussion and references). Recent works have pointed out that in general revision processes (including independent-learning) logit-response dynamics do  \emph{not} converge to Nash equilibria in potential games (\citet{CauDurGauTou14,AloNet15} discuss this issue explicitly).

\subsection{Preliminary Definitions}\label{sec:prely}
This section introduces the necessary tools from \cite{AloNet10}, while  the next subsection presents the revised \PoA and \PoS notions. Let $u_i(s_i,s_{-i})$ denote the utility of player $i$ when she chooses strategy $s_i$ and the other players' strategies are $s_{-i}$ (the latter is a vector of strategies). The combination of all these strategies is a state $s=(s_i,s_{-i})$. In logit-response dynamics, whenever a player is given the opportunity to revise her strategy, she will revise her current strategy with a strategy $s_i$ chosen according to the \emph{logit choice function} 
\begin{equation}
\label{eq:logit-choice}
p_i(s_i, s_{-i}) = \frac{e^{\beta u_i(s_i,s_{-i})}}{\sum_{s_i'\in S_i} e^{\beta u_i(s_i',s_{-i})}}
\end{equation}
where $s_{-i}$ is the current strategy profile (the strategies chosen by the others) and $\beta$ is a parameter representing the inverse level of noise in the players' decision ($0<\beta<\infty$), and $S_i$ is the set of strategies of player $i$.  The dynamics is defined as follows:
\begin{itemize}
	\item Select a subset $J$ of players according to a probability distribution $q$ defined on the subsets of players.
	\item Every selected player $i\in J$ chooses a new strategy $s_i$ with probability given by the logit choice function  \eqref{eq:logit-choice}.
\end{itemize}
Note that the dynamics depends on the \emph{revision process} $q$, and different revision processes yields different (logit-response) dynamics: 
\begin{itemize}
	\item Logit dynamics \cite{Blu93,Blu98} correspond to  \emph{asynchronous learning} in which only one player at the time is allowed to revise (i.e, $q(J)>0$ if and only if $J=\{i\}$ for some $i$). 
	\item In logit-response dynamics with \emph{independent learning} every subset of players has non-zero revising probability (i.e. $q(J)>0$ for every subset $J$). 
\end{itemize}

A state $s$ is \emph{stochastically stable} if it has positive probability in the limit invariant distribution as the noise goes to zero, 
\[
\lim_{\beta \rightarrow \infty} \mu^\beta(s)>0,
\] 
where $\mu^\beta$ is the stationary distribution of the dynamics with parameter $\beta$.\footnote{It is well known that the dynamics described above define an ergodic Markov chain over the possible strategy profiles for every $0 < \beta < \infty$. That is, the stationary distribution $\mu^{\beta}$ exists for all $\beta$.}

An edge $(s,s')$ is \emph{feasible} if $R_{s,s'}:=\{J| \ q(J)>0 \mbox{ and } s_k=s'_k \forall k\not \in J\}$ is non-empty (note that $R_{s,s'}$ consists of the subsets of players potentially leading from $s$ to $s'$). For any such feasible transition,  its \emph{waste} is defined as
\begin{align*}
	W_{s,s'} :=& \min_{J \in R_{s,s'}} W^{(J)}_{s,s'}, \mbox{ where} \\
	W^{(J)}_{s,s'}:=&\sum_{j\in J} \left( \max_{s_j''\in S_j} u_j(s''_j,s_{-j}) - u_j(s'_j,s_{-j})\right).
\end{align*}
The waste of a directed tree (or of path) $T$ is simply the sum of  the waste of its edges:
\[
W(T) :=\sum_{(s,s')\in T} W_{s,s'}.
\]
The \emph{stochastic potential} of a state $s$ is 
\[
W(s):=\min_{T\in \mathcal{T}(s)} W(T)
\]
where $\mathcal{T}(s)$ is the set of all directed trees $T$ in which every state $s''$ has a unique path to $s$.  That is, $T$ is a tree directed towards $s$ with each edge being feasible.

\begin{theorem}[\cite{AloNet10}]\label{th:stochastic-potential-minimizer}
	Consider the logit-response dynamics (with any revision process). A state is stochastically stable if and only if it minimizes $W(s)$ among all states.
\end{theorem}

It has been observed recently that, even for potential games\footnote{In this work we deal only with weighted potential games, that is, games that admit a vector $w$ and a potential function $\phi$ such that $u_i(s)-u_i(s')=(\phi(s')-\phi(s))w_i$ for all $i$ and for all $s$ and $s'$ that differ in exactly player $i$'s strategy.}, stochastically stable states are not necessarily Nash equilibria \cite{CauDurGauTou14,AloNet15}.

In our proofs we shall find it convenient to reason about the existence of \emph{zero-waste} paths between two states. Note that a zero-waste transition $(s,s')$  corresponds to a multiple best-response: Since $W_{s,s'}^{(J)}=0$ all players in $J$ 
simultaneously best-respond to the strategies in $s$ and the resulting profile is $s'$ (for $J=\{i\}$ we have a best-response).

\subsubsection{Revised Price of Anarchy and Stability}
We want to quantify the efficiency loss that occurs when we relax logit dynamics to  independent learning revision processes. Note that the set of stochastically stable states does not depend on the particular revision process $q$ satisfying the definition of independent learning (the stochastic potential depends only on the feasible transitions which do not change if $q$ and $q'$ are both independent learning revision processes). Given an instance $I$ (a game) of interests, and a corresponding  global cost function $cost$ mapping states into a nonnegative real, we define the following 
\[
\ind~\PoA(I) = \frac{\max_{s \in STABLE} cost(s,I)}{\min_s cost(s,I)}
\]
where  $STABLE$ denotes the  stochastically stable states for the logit-response dynamics with independent learning. By replacing `$\max$' with `$\min$' we obtain the notion  of $\ind~\PoS$. The original $\PoA$  and $\PoS$ are defined by replacing $STABLE$ with the set of all possible Nash equilibria. The restriction to logit dynamics, \asy-\PoA and \asy-\PoS, have $STABLE$ equal to the states minimizing the potential of the game. We generally consider classes of games and extend both definitions by taking the supremum over all instances: For a class of games $\mathcal C$ the $\ind~\PoA$ is defined as $\sup_{I\in \mathcal C} \ind~\PoA(I)$; $\ind~\PoA$ is defined as $\sup_{I\in \mathcal C} \ind~\PoS(I)$. 

The following easy result (whose elementary proof is given in Appendix~\ref{app:non-Nash}) is useful 
for games in which convergence to Nash equilibria is not guaranteed (like those considered here). 

\begin{theorem}\label{th:revPoS:general}
	For any game, the set of stochastically stable states in independent learning logit-response dynamics must contain at least one Nash equilibrium. Therefore $$\ind~\PoS \leq \PoA.$$
\end{theorem}

\section{Load Balancing Games}
We consider the load balancing problem which is defined as follows. There are $m$ identical machines and $n$ jobs, each job having its own weight. Each job is controlled by a player who can choose one machine (strategy). The strategies determine an allocation and the load for each machine (sum of the weights of jobs allocated to the machine). Naturally, each player aims at having her job on a machine with minimal load (so the payoff is defined as the negative of the load of the chosen machine).  The cost of an allocation is the so-called \emph{makespan}, that is, the maximum load over all machines.

\subsection{Independent Logit \PoA}

In this section we show that the \PoA in independent learning dynamics can be much worse than all prior bounds. In particular, it is higher than the classical \PoA which considers any Nash equilibria. This  means that there are some stochastically stable states which are \emph{not} Nash equilibria for this potential game. 

\begin{theorem}\label{th:ind-PoA}
	For load balancing games, 
	$$
	\ind~\PoA
	=m.
	$$
\end{theorem}

\begin{proof}
	Since the upper bound is trivial, we shall prove only the lower bound.
	Consider the instance with $m$ machines and $n=lm-1$ identical jobs (say of size $1$), for a large positive integer $l$. The optimum has cost $l$ and it is given by any allocation in which one machine has $l-1$ jobs, and all other machines have $l$ jobs. 
	
	We claim that \emph{every} state is stochastically stable, including those in which \emph{all} jobs go to the same machine (and whose cost is thus  $lm-1$).  To prove our claim, we observe that in every state $s$ there is at least one machine $m_{high}$ with at least $l$ jobs and another machine $m_{low}$ with at most $l-1$ jobs. This implies that we can construct a zero-waste path from any state $s$ to every other state $s'$ as follows:
	\begin{enumerate}
		\item First all jobs in  machine $m_{high}$ move to  machine $m_{low}$;
		\item Now that machine $m_{high}$ is empty, all jobs move there;
		\label{itm:move-to-empty}
		\item Then from there all jobs move to the machine they occupy in $s'$ (some do not move because they are already in the correct machine). \label{itm:move-to-destination}
	\end{enumerate} 
	Note that each move is a best-response made simultaneously by a subset $J$ of players (because players move to a machine with strictly lower load). Therefore each of the three transitions has zero waste.  Since the stochastically stable state with all jobs in the same machine costs $lm-1$, and the optimum is $l$, we have shown that $$\ind~\PoA \geq m - 1/l$$ for all $l$. The theorem thus follows by taking $l$ large enough.  
	 \end{proof}

\subsection{Independent Logit \PoS}
In this section we show a tight bound on the \ind~\PoS. 

The strategy to prove that a certain state is \emph{not} stochastically stable is to prove that it is ``easy'' to leave this state. That is, there exists a tree of small waste which is directed \emph{from} this state to all other states. The next lemma says that such a tree can be obtained by finding a small-waste path leading to a state in which 
one machine has \emph{zero load}. 

\begin{lemma}\label{le:universal-zero-paths}
	Let $s$ be an allocation in which one machine has zero load. Then there exists a zero-waste path from $s$ to every other state.   
\end{lemma}
The proof of this lemma is exactly the same argument used in the proof of Theorem~\ref{th:ind-PoA} above (Items~\ref{itm:move-to-empty}-\ref{itm:move-to-destination}).  Next, we  introduce some additional notation that will ease the proofs below.

\paragraph*{Compact notation for states and jobs allocations.}
Given a set of jobs weights, we partition these weights among the machines, without distinguishing among the machines and among players having identical weights. This results in a \emph{class} of states having the same cost. For instance, for jobs of size $\{2,1,1,1\}$ the partition
\[
[2],[1,1],[1]
\]
denotes a \emph{class} of allocations in which one machine gets a job of size $2$, does not matter which of the three machines, another machine gets two out of the three jobs of size $1$, and the remaining machine gets the remaining job of size $1$. An explicit representation would be in terms of which player goes to which machine: for players $\{p_1,p_2,p_3,p_4\}=\{2,1,1,1\}$ the states 
\begin{align*}
	[p_1],[p_2,p_3],[p_4],& & [p_1],[p_2,p_4],[p_3],& & [p_1],[p_3,p_4],[p_1],\ldots& \ldots [p_4],[p_2,p_3],[p_1]
\end{align*}
correspond to the partition $[2],[1,1],[1]$. Finally, we call a job of size $s$ an $s$-job. \qed

\begin{theorem}\label{th:lower-bound-PoS}
	For load balancing games,
	$$
	\ind~\PoS \geq 2\Big(1 - \frac{1}{m+1}\Big).
	$$
\end{theorem}

\begin{proof}
	Consider instances with $m$ machines and jobs
	\[
	\{\Delta - \delta, \Delta - \delta, \Deltajobs{m-2}, \deltajobs{lm}\} 
	\]
	where $l$ is a (sufficiently large) positive integer and
	\begin{align*}
		\Delta =& \frac{m}{m+1}, & \mbox{ and }& & \delta =& \frac{\Delta}{lm}.
	\end{align*}
	First, optimal allocations have cost $1$ and can be obtained (only) by allocating one large job (size $\Delta$ or $\Delta-\delta$) in each machine as follows:
	\begin{eqnarray}
	\label{eq:opt} 
	[\Delta - \delta, \deltajobs l],[\Delta - \delta, \deltajobs l],[\Delta,\deltajobs l],\ldots,[\Delta,\deltajobs l]. 
	\end{eqnarray}
	(Note that $\Delta + l \delta =1$ and the cost is indeed $1$.)
	
	\noindent
	Second, there are Nash equilibria of cost $2(\Delta -\delta)$ which are of the form
	\begin{equation}
	{[\Delta - \delta, \Delta - \delta], [\delta,\ldots,\delta],[\Delta], \ldots, [\Delta].}
	\label{eq:apx}
	\end{equation}
	We shall prove that only allocations of cost $2(\Delta -\delta)$ are stochastically stable, thus implying 
	\[
	\ind~\PoS \geq 2(\Delta - \delta) = 2\Big(1 - \frac{1}{m+1}-\delta\Big).
	\] 
	The theorem thus follows by taking $\delta$ arbitrarily small, i.e., $l$ arbitrarily large. 
	
	\begin{claim}\label{cla:opt-no-stable}
		Every allocation with one large job per machine has a zero-waste path to every other state. 
	\end{claim}
	\begin{proof}[Proof of Claim] In any allocation that allocates one large job per machine, there exists a machine with load at most $\Delta - \delta+l\delta=1-\delta$. The following sequence of best-response moves leads to another allocation with one machine having zero load:
		\begin{align*}
			[\Delta - \delta, \deltajobs l]&, & [\Delta - \delta, \deltajobs l]&, & [\Delta,\deltajobs l], \ldots,& & [\Delta,\deltajobs l]&
			& \rightarrow& \\ 
			[\Delta - \delta, \deltajobs l]&, & [\Delta - \delta, \deltajobs{l(m-1)}]&, & [\Delta],\ldots,& & [\Delta]& & \rightarrow& \\ 
			[\Delta - \delta]&, & [\Delta - \delta]&, & [\Delta],\ldots&, & [\Delta,\deltajobs{lm}]& & \rightarrow& \\ 
			[\Delta - \delta, \Delta, \deltajobs{lm}]&, & [\Delta - \delta]&, & [\Delta], \ldots&, & [\textbf{empty}]& &
		\end{align*}
		The claim then follows by Lemma~\ref{le:universal-zero-paths}.
		 \end{proof}

	\begin{claim}\label{cla:apx-delta-paths} 
		Every Nash equilibrium of cost $2(\Delta-\delta)$  requires a minimal waste of $\delta$ and 
		there exists a path of waste $\delta$ from such an equilibrium to every other state.
	\end{claim}
	\begin{proof}[Proof of Claim]
		Every Nash equilibrium of cost $2(\Delta-\delta)$ must allocate two $(\Delta-\delta)$-jobs to the same machine and thus it must be of the form:
		\[
		[\Delta - \delta, \Delta - \delta], [\delta,\ldots,\delta],[\Delta], \ldots, [\Delta]. 
		\]
		Since the machine containing all $\delta$-jobs has load $\Delta$, every move of one (or more) job(s) requires a waste of $\delta$ or larger. To get a path of waste $\delta$, consider the following transitions:
		\begin{align*}
			[\Delta - \delta, \Delta - \delta]&, & [\delta,\ldots,\delta]&, & [\Delta], \ldots,& [\Delta]
			& \stackrel{(waste=\delta)}{\rightarrow}& \\ 
			[\Delta - \delta]&, & [\Delta - \delta, \delta,\ldots,\delta]&, & [\Delta], \ldots,& [\Delta]
			& \stackrel{(waste=0)}{\rightarrow}& \\ 
			[\Delta - \delta, \Delta - \delta, \delta,\ldots,\delta]&, & [\textbf{empty}]&, & [\Delta], \ldots,& [\Delta]
		\end{align*}
		and apply Lemma~\ref{le:universal-zero-paths} to obtain the rest of the path.
		\end{proof}
	
	To prove the theorem, let  $APX$ denote the approximate allocations of the form \eqref{eq:apx}. We use the two claims above to prove that for every $s \not \in APX$
	\begin{align}
		\label{eq:OPT-waste}
		W(s) & \geq |APX|\delta \\
		\intertext{while for every $s_{apx}\in APX$}
		W(s_{apx}) & \leq (|APX|-1)\delta
		\label{eq:APX-waste}
	\end{align}
	Inequality \eqref{eq:OPT-waste} follows from the fact that every state in $APX$ requires a waste of $\delta$ to leave, thus any tree towards $s$ must have waste $|APX|\delta$. For  \eqref{eq:APX-waste}, we argue as follows. First we build a partial tree from the $|APX|-1$ states in $APX$ to $s_{apx}$. Second, we observe that every non-Nash equilibrium has a zero-waste path to some Nash equilibrium. Third, we show that set of all Nash equilibria is the union of $APX$ and $OPT$, where $OPT$ are the allocations of cost $1$ (details in Appendix~\ref{sec:proof:th:lower-bound-PoS}). Then, by Claim~\ref{cla:opt-no-stable}, the paths from $OPT$ to $s_{apx}$ have zero waste.
	 \end{proof}

A matching upper bound follows from the known upper bound on \PoA (see Table~\ref{tbl:results}) and the fact that  $\ind~\PoS \leq \PoA$ (Theorem~\ref{th:revPoS:general}).

\begin{corollary}\label{cor:independent-bounds}
	For load balancing games,
	$$
	\ind~\PoS = 2\Big(1 - \frac{1}{m+1}\Big).
	$$
\end{corollary}

\section{Broadcast Network Design Games}
In this section we consider broadcast network design games studied in  \cite{KawMak13,MamMih15}. In these games, we are given a network in which edges have costs, and a set of players  located in some of the nodes. The players aim at connecting to a common terminal node $t$ and their strategy is to choose which route to use. The cost of each edge is then \emph{shared equally} among all players using that edge: if $n_e$ players select a route containing edge $e$ then each of them pays a share $c_e/n_e$ of the cost $c_e$ of the edge.  The cost for a player is the sum of the shares of the edges contained in the chosen route, and the social cost is the sum of the costs of the used edges.

Figure~\ref{fig:network-design:triangle}
shows an example from \cite{KawMak13} in which $s_1$ and $s_2$ denote the two players (their locations). We use this instance to show that also in these games synchronization is important for the quality of equilibria.

\begin{figure}\centering
	\includegraphics[scale=.7]{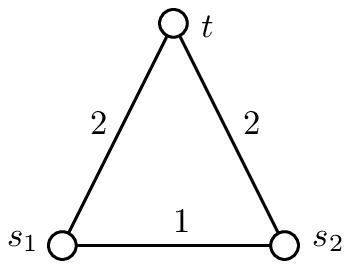}
	\includegraphics{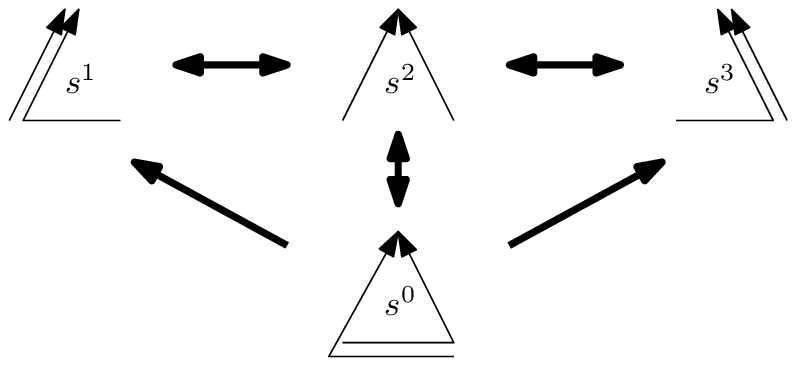}
	\caption{On the left, an instance of broadcast network design game for which a non-Nash equilibrium is also stochastically stable. On the right, the possible states for this instance (think arrows represent zero-waste transitions).}
	\label{fig:network-design:triangle}
	\label{fig:network-dsign:non-Nash-stable}
\end{figure}

\begin{theorem}\label{th:network-design-games}
	There exists an instance $I$ of a broadcast network design game for which $\PoA(I)=4/3$ but $$
	\ind~\PoA(I) = 5/3.
	$$
	In particular, some non-Nash equilibrium is stochastically stable.
\end{theorem}

\begin{proof}
	Consider the instance in Figure~\ref{fig:network-design:triangle}. Each player has two strategies and therefore there are only four possible states shown in Figure~\ref{fig:network-dsign:non-Nash-stable}. Note that the Nash equilibria are  $s^1, s^2$ and $s^3$, and $s^1$ and $s^3$ can be obtained from $s^2$ by a move of a single player whose cost remains $2$ in both cases (thus the move is a best-response). Therefore, transitions $(s^2,s^1)$ and  $(s^2,s^3)$ are zero-waste transitions. Moreover, starting from $s^2$, if both players move we obtain state $s^0$ and again a zero-waste transition $(s^2,s^0)$ (being each move a best-response). Finally, transitions from $s^0$ to any Nash equilibrium are zero-waste because the player(s) changing strategy actually improve. Since all transitions are zero-waste, we have $W(s)=0$ for all four states and therefore $s^0$ is stochastically stable. Since the cost of $s^0$ is $5$ and the optimum is $3$, the theorem follows.
	 \end{proof}

We next consider broadcast network design games \emph{over parallel links}: All players are located on \emph{the same node} $s$ and the graph consists of $m$ parallel links (representing resources of different costs) connecting to the common target $t$. 
Let  $\ell_1 \leq \ell_2 \leq \cdots \leq \ell_m$ be the costs of the links, and let $N_k$ denote the state in which all players choose the $k^{th}$ link in the order above (in particular $N_1$ is optimal and also a Nash equilibrium).  
The following  result provides sufficient conditions for which the dynamics will select only the optimum.

\begin{theorem}\label{th:broadcast:identical-locations}
	For network broadcast games over parallel links, if there is a unique shortest link ($\ell_1<\ell_2$) or all links have the same cost then 
	$$\ind~\PoA=1.$$
\end{theorem}

The proof uses the following \emph{radius-coradius} argument from \cite{AloNet10}. The \emph{basin of attraction} of a state $s$ is the set $B(s)$ of all states $s'$ such that there is a zero-waste path from $s$ to $s'$. The \emph{radius} $R(s)$ is the minimum number $r$ such that there is a path of waste $r$ from $s$ to some $s'\not \in B(s)$. The \emph{coradius} $CR(s)$ is the minimum number $cr$ such that, for every $s''\not \in R(s)$,  there is a path of waste at most $cr$ from $s''$ to $s$. 

\begin{lemma}[\citet{AloNet10}]\label{le:radius-coradius}
	If $R(s)>CR(s)$, then the set of stochastically stable states are exactly those in $L(s)$, where $L(s)$ is the set of $s'\in B(s)$ for which $s \in B(s')$.
\end{lemma}

We next apply this result to our problem. 
\begin{lemma}\label{le:radiou-coradius-NDE-game}
	$R(N_1)-CR(N_1) = (\ell_2-\ell_1)H(n)$, where $H(n)=1+1/2+\cdots+1/n$ is the harmonic function. 
\end{lemma}

\begin{proof}
	Let $b_1$ be the smallest positive integer such that if link~1 is chosen by at least $b_1$ players, then moving to link~1 is a best-response for all other players. That is
	\[
	\frac{\ell_1}{b_1+1} \leq \frac{\ell_2}{n-b_1}.
	\] 
	\begin{claim}
		$
		R(N_1)\geq \left(\ell_2-\frac{\ell_1}{n}\right) + \left(\frac{\ell_2}{2}-\frac{\ell_1}{n-1}\right) + \cdots + \left(\frac{\ell_2}{n-b_1}-\frac{\ell_1}{b_1-1}\right).
		$
	\end{claim}
	\begin{proof}[Proof of Claim]
		By definition of $b_1$, in order to reach a state not in  $B(N_1)$ we have to move $n-b_1+1$ players from link~1. The minimal waste path for achieving this is to move these players one by one and all to link~2. Note that the $k^{th}$ move has waste $\frac{\ell_2}{k}-\frac{\ell_1}{n-k+2}$ and the last move ($k=n-b_1+1$) has zero waste by definition of $b_1$: indeed $\frac{\ell_1}{b_1}> \frac{\ell_2}{n-b_1+1}$. 
	\end{proof}
	
	\begin{claim}
		$
		CR(N_1)\leq \left(\ell_1-\frac{\ell_2}{n}\right) + \left(\frac{\ell_1}{2}-\frac{\ell_2}{n-1}\right) + \cdots + \left(\frac{\ell_1}{b_1}-\frac{\ell_2}{n-b_1+1}\right)
		$
	\end{claim}
	\begin{proof}[Proof of Claim]
		By definition of $b_1$, starting from any state, it is enough to move at most $b_1$ players into link~1 to obtain a state is in $B(N_1)$. We move players one by one. Suppose first that initially $\ell_1$ contains no players. Then the $k^{th}$ move has waste $\frac{\ell_1}{k} - \frac{\ell_{j_k}}{n_{j_k}}$, where $\ell_{j_k}\neq \ell_1$ is the link of the $k^{th}$ player we move, and $n_{j_k}$ is the number of players currently in that link. Since $\ell_{j_k}\geq \ell_2$ and $n_{j_k} \leq n - k +1$ (link $1$ contains $k-1$ before this move is done), the waste of each move is upper bounded by $\frac{\ell_1}{k} - \frac{\ell_2}{n-k+1}$ and a total of $b_1$ moves is enough. This gives the upper bound in the claim. Finally, if $\ell_1$ already contained some players in the initial state, then an even smaller waste is sufficient (it suffice to consider the first $n_1$ moves above as already done, where $n_1$ is the initial number of players in $\ell_1$). 
	\end{proof}
	By combining the two claims above and rearranging terms, we get $R(N_1)-CR(N_1) \geq \ell_2 + \frac{\ell_2}{2} + \cdots + \frac{\ell_2}{n} -(\ell_1 + \frac{\ell_1}{2} + \cdots + \frac{\ell_1}{n})$, which proves the lemma.
\end{proof}

\begin{proof}[Proof of Theorem~\ref{th:broadcast:identical-locations}]
	If all links have the same cost ($\ell_1=\ell_m$) then every state $N_k$ is optimal. Since these states are also strict Nash equilibrium, the convergence result in \cite{AloNet15} says that the stochastically stable states are contained in these strict Nash equilibria (formally, the game is weakly acyclic since every other state is not a Nash equilibrium).

	Now consider the case in which there is only \emph{one} optimal link ($\ell_1<\ell_2$), then we have $R(N_1)>CR(N_1)$ by Lemma~\ref{le:radiou-coradius-NDE-game}. Moreover $L(N_1)=\{N_1\}$ since this is a strict Nash equilibrium. We can thus apply Lemma~\ref{le:radius-coradius}.
	\end{proof}

\section{Conclusion and Open Questions}
In this work we have studied a refinement of the \PoA and \PoS based on the long-run equilibria in a general noisy best-response dynamics (a generalization of logit dynamics). The analysis is based on the characterization by \cite{AloNet10} and it provides new insights on the importance of the order in which players revise their strategies. 
On the one hand, restricting the analysis only to potential minimizers may  not be significant for other natural dynamics of the same type. On the other hand, a comparison between the two dynamics is useful to understand the importance of synchronous moves.

We remark that, in load balancing games,  the bounds for independent learning logit-response dynamics are in a sense the \emph{worst possible}.  First, $m$ is an upper bound to \emph{any} configuration, and therefore $\ind~\PoA\leq m$. Second, in every game we have $\ind~\PoS \leq \PoA$ (see Theorem~\ref{th:revPoS:general}). In load balancing games both inequalities are tight.

A natural alternative would be to consider the expected cost at stationary distribution \cite{AulFerPasPer13}. Here the main technical difficulty lies in obtaining the stationary distribution (see e.g. \cite{AulFerPasPenPer13} for such an issue). 
The advantage of stochastic stability is that it avoids an explicit computation of the stationary distribution.  However, the resulting  bounds can be considered too pessimistic in some cases. For instance, when each player is selected for revising her strategy with probability $p$ sufficiently small, the dynamics should be close to the logit dynamics. \citet{MarSha12} suggest to reduce to potential minimizers, that is, to show that stochastically stable states are potential minimizers also for logit-response dynamics with more general revision processes. We note that this approach requires additional hypothesis on the game (or on the players behavior). In particular, one should assume that players have a unique best response to have convergence to Nash equilibria \cite{CauDurGauTou14,AloNet15}. In this case, one gets $\ind~\PoA \leq \PoA$, and in some cases the results of \cite{MarSha12} say that $\ind~\PoA \leq \asy~\PoA$. 

The study of other problems considered in \cite{AsaSab09,KawMak13,MamMih15} is an interesting direction for future research. Some of these bounds on potential minimizers are already quite sophisticated, and stochastic stability in independent learning  may require non-trivial analysis.
Estimating \ind~\PoA and \PoS is still open even for ring networks \cite{MamMih15}.

\paragraph{Acknowledgments.} I am grateful to Francesco Pasquale for comments on an earlier version of this work, and to Carlos Alos-Ferrer for bringing \cite{CauDurGauTou14,AloNet15} to my attention.

\appendix
\section{Postponed Proofs and Additional Details}

\subsection{Proof of Theorem~\ref{th:revPoS:general}}\label{app:non-Nash}
	Observe that for any state $\bar s$ which is not a Nash equilibrium there exists a sequence of best-response moves that reaches some Nash equilibrium $s$. 
	That is, there exists a zero-waste path from any non-Nash $\bar s$ to some Nash $s$. 
	
	This implies that $W(\bar s) \geq W(s)$ because any waste tree directed to $\bar s$ can be transformed into a tree of no larger waste directed into $s$ (by adding the path mentioned above and removing other transitions). This shows that the waste $W(\bar s)$ cannot be smaller than the waste of all Nash equilibria.

\subsection{Postponed Details for the Proof of Theorem~\ref{th:lower-bound-PoS}}\label{sec:proof:th:lower-bound-PoS}
We need to prove that the set of all Nash equilibria is just the union of $OPT$ and $APX$, that is, the allocations of the form  \eqref{eq:opt} and \eqref{eq:apx}, respectively.

We first show that any allocation with a machine containing a $\Delta$-job and another large job (size $\Delta$ or $\Delta - \delta$) cannot be a Nash equilibrium. To be a Nash equilibrium, the load of every other machine should be at least $\Delta$. This will not be possible given the remaining jobs: If we have a machine with a $\Delta$-job and $\Delta-\delta$-job, then the  jobs that are left for the other $m-1$ machines are 
\[
\{\Delta - \delta, \Deltajobs{m-3}, \deltajobs{lm}\} 
\]
where the  $\delta$-jobs sum up to $\Delta$. Thus the total load we can distribute is $(m-1)\Delta - \delta$, and some machine will have load less than $\Delta$.  The same holds if we start with a machine having two $\Delta$-jobs. 

Finally, we  show that every allocation with one large job per machine has a zero-waste path to some allocation in $OPT$. Indeed, allocations with one large job per machine are of the form
\newcommand{\deltajobspar}[1]{\underbrace{\delta,\ldots,\delta}_{#1}}
\begin{eqnarray}
\label{eq:one-large-per-machine} 
[\Delta - \delta, \deltajobspar{k_1} ],[\Delta - \delta, \deltajobspar{k_2}],[\Delta,\deltajobspar{k_3}],\ldots,[\Delta,\deltajobspar{k_m}]. 
\end{eqnarray}
To obtain an allocation in $OPT$ it is enough to move the $\delta$-jobs from higher load machines to lower load machines (or to machines with the same load), until we obtain  $k_1=\cdots=k_m=l$, that is \eqref{eq:opt}.

\bibliographystyle{alpha}
\bibliography{stablePoA}

\end{document}